\documentclass[a4paper,UKenglish]{amsart}
\usepackage{url} 
 
\usepackage{microtype}

\newtheorem{proposition}{Proposition}

\title[the Montagovian generative lexicon]{A type theoretical framework\\ for natural language semantics:\\  the Montagovian generative lexicon}

\author
{Christian Retor\'e} 
\thanks{\textbf{Affiliation:} LaBRI, Universit\'e de Bordeaux (\& IRIT-CNRS, Toulouse)  ---- This work  was achieved during my CNRS-sabbatical  at IRIT and has been supported by the project ANR LOCI}
\address{LaBRI, Universit\'e de Bordeaux (\& IRIT-CNRS, Toulouse)}
\email{christian.retore@labri.fr}
\urladdr{http://www.labri.fr/perso/retore}


\usepackage{prooftree} 
\newcommand\coercion[2]{\mathsf c_{#1\fl#2}} 
 
\newcommand\chk{\coercion{h}{k}}
\newcommand\cij{\coercion{i}{j}}
\newcommand\cik{\coercion{i}{k}}

\newcommand\cjk{\coercion{j}{k}}
\newcommand\cjl{\coercion{j}{l}}
\newcommand\ckl{\coercion{k}{l}}

\newcommand\ckj{\coercion{k}{j}}

\usepackage{graphicx}
\DeclareGraphicsExtensions{.pdf}

\newtheorem{prop}{Property}
\newtheorem{term}{Term}

\newcommand\Land{\&^\Pi}

\newcommand\pred[1]{\widehat{#1}}
\newcommand\flex{ \textsc{(f)}}  
\newcommand\rig{ \textsc{(r)}} 

\newcommand\type[1]{^{#1}}
\newcommand\et\land
\newcommand\fl{\rightarrow}

\usepackage{latexsym} 
\usepackage{amssymb} 

\newcommand{\ltyn}{\ensuremath{\Lambda\mathsf{Ty}_n}}
\newcommand\ma[1]{"\emph{#1}"}

\newcommand\ttt{\mathbf{t}}
\newcommand\eee{\mathbf{e}}

\newcommand\systF{\mbox{\textsf{F} }}

\usepackage{gb4e} 

\begin{document}

\maketitle 

\begin{abstract}
We present a framework, named the Montagovian generative lexicon,  for computing the semantics of natural language sentences, expressed in many sorted  higher order logic.  Word meaning 
is depicted by lambda terms of second order lambda calculus (Girard's system F) with base types including a type for propositions and many types for sorts of a many sorted logic.  This framework is able to integrate a proper treatment of lexical phenomena into a Montagovian  compositional semantics, including the restriction of selection which imposes the nature of the arguments of a predicate, and the possible adaptation of a word meaning to some contexts. Among these adaptations of a word's sense to the context,  ontological inclusions are handled by an extension of system F with coercive subtyping that is introduced in the present paper. The benefits of this framework for lexical pragmatics are illustrated on meaning transfers and coercions, on possible and impossible copredication over different senses, on deverbal ambiguities, and on "fictive motion". Next we show that the compositional treatment of determiners, quantifiers, plurals,... are finer grained  in our framework. We then conclude with the  linguistic, logical and computational  perspectives opened by the Montagovian generative lexicon. 
\end{abstract}

\section{Introduction: word meaning and compositional semantics}

The study of natural language semantic and its automated analysis  is usually divided into \emph{formal semantics}, usually compositional, 
which has strong connections with logic and with philosophy of language,
and  \emph{lexical semantics}  which rather concerns word meaning and their interrelations, derivational morphology and knowledge representation. 
Roughly speaking, given an utterance,  formal semantics tries to determine \emph{who does what} according to this utterance, while lexical semantics analyses the concepts under discussions and their interplay i.e. 
 \emph{what it speaks about}.  
 \begin{exe} 
 \ex \emph{A sentence:} Some club defeated Leeds. 
 \ex \emph{Its formal semantics:} $\exists x:\eee\  (\texttt{club}(x)\ \et\ \mathit{defeated}(x,Leeds))$
 \ex 
\emph{Lexical semantics of the verb as found in a dictionary:} defeat:  
 \begin{xlist}  
 \ex 
 overcome in a contest, election, battle, etc.; prevail over; vanquish
\ex  to frustrate; thwart.
\ex  to eliminate or deprive of something expected
\end{xlist} 
\end{exe} 

Although any applications in computational linguistics requires both 
formal and compositional semantics rather applies in man machine dialogue, text generation and lexical semantics in information retrieval and classification. 
Herein we shall endow compositional semantics with a treatment of some of lexical semantics issues,
in particular for picking up the right word sense in a given context. Of course any sensible analyser, including human beings, or Moot's Grail parser \cite{moot10grail}
combines both the predicate argument structures and the relations between lexical meanings to build a semantic representation and to understand the utterance. 

\subsection{The syntax of compositional semantics}

As opposed to many 
contributions to the domain of linguistic known as "formal semantics"  
the present paper neither deals  with reference nor with truth in a given situation: we only build a logical formulae
first order or higher order, single or many sorted) that can be thereafter interpreted as one wants, if he wishes to. Hence are not committed to any particular kind of interpretation like truth values, possible worlds, game semantics,...

In the traditional view as exposed by Montague, the process of semantic interpretation  of a sentence, consists in computing a logical formula including logical modalities and intensional operators, 
from syntax and word meanings, 
and to interpret it in possible world semantics. 
Although Montague thought that intermediate steps were meaningless and should be wiped off just after computing truth values and references, in this paper we precisely focus on the intermediate step, the logical formula, that can be called the \emph{logical form} of the sentence, with particular attention to the way it is computed --- for the time being, we leave out the interpretation of these formulae.  
A reason for doing so is that we can encompass subtle questions, 
like vague predicates, generalised and vague quantifiers, for which standard notions of truth and references are inadequate possibly some interactive interpretation would be better suited, e.g.  like 
\cite{AbrusciRetoreCLMPS,LQ2012}. Another reason is that, apart from these difficult questions, 
we do not have modification to bring to standard interpretations. 

\subsection{Brief reminder on Montague semantics} 
\label{Montague} 

Let us  briefly remind the reader how one computes the logical form according to the montagovian view. 
Assume for simplicity that a  syntactic analysis is a tree  specifying for each node, which subtree applies to the other one --- the one that is applied is called the function while the other is called its argument.  A semantic lexicon provides a simply typed $\lambda$-term $[w]$ for each word $w$.  The semantics of a leaf (hence a word) $w$ is $[w]$ and the semantic $[t]$ of a sub syntactic tree 
$t=(t_1, t_2)$ is recursively defined as $[t]= ([t_1]\ [t_2])$ that is $[t_1]$ \emph{applied to} $[t_2]$, if  $[t_1]$ is the function and $[t_2]$ the argument  --- and as $[t]= ([t_2]\ [t_1])$ otherwise, i.e. when $[t_2]$ is the function and $[t_1]$ the argument. 

The typed $\lambda$-terms from the lexicon are given in such a way that the function always has a semantic type 
of the shape $a\fl b$ that matches the type $a$ of the argument, and the semantics associated with the whole tree has the semantic type $\ttt$, that is the type of propositions. This correspondence between syntactical categories and semantic types, 
which extends into a correspondence between parse structures and logical forms is crystal clear in categorial grammars, see e.g. \cite[Chapter 3]{MootRetore2012lcg}. Typed $\lambda$-terms usually are defined out of two  base types, $\eee$ for individuals  (also known as \textbf{e}ntities) and  $\ttt$ for propositions (which have a \textbf{t}ruth value). Logical formulae can be defined in this typed $\lambda$-calculus as first observed by Church long ago. 
This early use of lambda calculus, where formulae are viewed typed lambda terms, can not be merged with the more familiar view of  typed lambda terms as proofs. The proof which such a typed lambda term correspond to is simply the proof that the formula is well formed, e.g. that a two-place predicate is properly applied to \emph{two} individual terms of type $\eee$ and not to more or less objects, nor to objects of a different type etc. 
This initial vision of lambda calculus was designed for a proper handling of substitution in deductive systems \`a la Hilbert. 
One needs constants for the logical quantifiers and connectives:

\begin{center}
$
\begin{array}[t]{r|l}
\mbox{Constant} & \mbox{Type}\\ \hline 
	\exists & (\eee \fl \ttt) \fl \ttt \\ 
	\forall & (\eee \fl \ttt) \fl \ttt \\ 
\end{array} 
$ \hfill 
$
\begin{array}[t]{r|l} 
\mbox{Constant} & \mbox{Type}\\ \hline 
	\textrm{and} & \ttt \fl (\ttt \fl \ttt) \\ 
	\textrm{or} & \ttt \fl (\ttt \fl \ttt) \\ 
	\textrm{implies} & \ttt \fl (\ttt \fl \ttt)
\end{array} 
$ \hfill 
$
\begin{array}[t]{r|l} 
\mbox{Constant} & \mbox{Type}\\ \hline 
	\mathit{defeated} & \eee \fl (\eee \fl \ttt) \\ 
	\mathit{won, voted} & (\eee \fl \ttt) \\ 
	\mathit{Liverpool, Leeds} & \eee 
\end{array} 
$
\end{center} 

\noindent as well as predicates for the precise language to be described --- 
a binary predicate like $won$ has the type $\eee\fl\eee\fl\ttt$.

A small example goes as follows. Assume the syntax says that the structure of the sentence "\emph{Some club defeated Leeds.}" is 
\begin{center}
(some\ (club)) (defeated\ Leeds) 
\end{center} 
where the function is always the term on the left. If the semantic terms are as in the lexicon in figure \ref{semanticlexicon}, placing the semantical terms in place of the words yields a large $\lambda$-term that can be reduced: 

\begin{figure} 
\begin{center} 
\begin{tabular}{ll} \hline 
\textbf{word} &  \textbf{\itshape semantic type $u^*$}\\ 
& \textbf{\itshape  semantics~: $\lambda$-term of type $u^*$}\\ 
&  {\itshape  $x\type{v}$ the variable or constant $x$ 
is of type $v$}\\ \hline 
\textit{some} 
& $(e\fl t)\fl ((e\fl t) \fl t)$\\ 
& $\lambda P\type{e\fl t}\  \lambda Q\type{e\fl t}\  
(\exists\type{(e\fl t)\fl t}\  (\lambda x\type{e}  (\et\type{t\fl (t\fl t)} (P\ x) (Q\ x))))$ \\  \hline 
\textit{club}  & $e\fl t$\\ 
& $\lambda x\type{e} (\texttt{club}\type{e\fl t}\  x)$\\  \hline 
\textit{defeated} & $e\fl (e \fl t)$\\ 
& $\lambda y\type{e}\  \lambda x\type{e}\  ((\texttt{speak\_about}\type{e \fl (e \fl t)}\  x)  y)$ \\  \hline 
\textit{Leeds} &$e$ \\ &  Leeds 
\end{tabular}
\end{center} 
\caption{A simple semantic lexicon} 
\label{semanticlexicon}
\end{figure}

$$
\begin{array}{c} 
\Big(\big(\lambda P\type{e\fl t}\ \lambda Q\type{e\fl t}\  (\exists\type{(e\fl t)\fl t}\  (\lambda x\type{e}  (\et (P\ x) (Q\ x))))\big)
\big(\lambda x\type{e} (\texttt{club}\type{e\fl t}\  x)\big)\Big) \\ 
\Big(
\big(\lambda y\type{e}\  \lambda x\type{e}\  ((\texttt{defeated}\type{e\fl (e\fl t)}\  x)  y)\big)\ Leeds\type{e}\Big)\\ 
\multicolumn{1}{c}{\downarrow \beta}\\ 
\big(\lambda Q\type{e\fl t}\  (\exists\type{(e\fl t)\fl t}\  (\lambda x\type{e}  (\et\type{t\fl (t\fl t)}  
(\texttt{club}\type{e\fl t}\  x) (Q\ x))))\big)\\ 
\big(\lambda x\type{e} \ ((\texttt{defeated}\type{e\fl (e \fl t)}\  x)  Leeds\type{e})\big)\\ 
\multicolumn{1}{c}{\downarrow \beta}\\ 
\big(\exists\type{(e\fl t)\fl t}\  (\lambda x\type{e}  (\et (\texttt{club}\type{e\fl t}\  x) ((\texttt{defeated}\type{e\fl (e\fl t)}\  x)  Leeds\type{e})))\big)
\end{array}
$$ 

This $\lambda$-term of type $\ttt$ that can be called the \emph{logical form} of the sentence, represents the following formula of predicate 
calculus (admittedly more pleasant to read): 

$$\exists x:\eee\  (\texttt{club}(x)\ \et\ \mathit{defeated}(x,Leeds))$$

The above described procedure is quite general: starting  a properly defined semantic lexicon whose terms only contains the logical constants and the predicates of the given language one always obtain a logical formula. Indeed, such $\lambda$-terms always reduce to a unique normal form and any normal $\lambda$-term of type $\ttt$ (preferably $\eta$ long, see e.g. \cite[Chapter 3]{MootRetore2012lcg}) corresponds to a logical formula. 

If we closely look at the Montagovian setting described above,  
we observe that it is weaving  two different "logics": 

\begin{description} 
\item[Logic/calculus for meaning assembly] (a.k.a glue logic, metalogic,...) In our example, this is simply typed $\lambda$-calculus with two base types $\eee$ and $\ttt$ --- these terms are the proof in intuitionistic propositional logic. 
\item[Logic/language for semantic representations] In our example, that is higher-order predicate logic.\footnote{It can be first-order logic if reification is used, but this may 
induce unnatural structure and exclude some readings.}  
\end{description} 

The framework we present in this paper mainly concerns the extension of the metalogic and the reorganisation of the lexicon in order to incorporate some phenomena of lexical semantics, first of all restrictions of selection. Indeed, in the standard type system above nothing prevents a mismatch between the real nature of the argument and its expected nature. 
Consider the following sentences:\footnote{We use the standard linguistic notation: a \ma{*} in front of a sentence points out that the sentence is incorrect, a \ma{?} indicates that the correctness can be discussed and the absence of any symbol in front means that the sentence is correct.} 

\begin{exe} 
\ex * A chair barks. 
\ex * Jim ate a departure 
\ex ?  The five is fast \label{five} 
\end{exe} 

Although they can be syntactically analysed, they should not receive a
semantical analysis. Indeed, \ma{barks} requires a \ma{dog} 
or at least an \ma{animate} subject while a \ma{chair} is neither of them; 
\ma{departure} is an event, which cannot be an \ma{inanimate}  object that could  be eaten; 
finally a \ma{number} like \ma{five} cannot do anything fast --- but there are particular contexts in which this can happen and we shall also handles these meaning transfers.

\subsection{The need  of integrating lexical semantics in formal semantics} 

In order to block the interpretation of the semantically illformed sentences above, it is quite natural to use \emph{types}, where the word \emph{type}  be both understood in its intuitive and in its formal meaning. The type of the subject of barks should be \ma{dog}, the type of \ma{fast} objects should be \ma{animate}, and the type of the object of \ma{ate} should be \ma{inanimate}. Clearly, having, on the formal side a unique type $\eee$ for all entities is not sufficient.

The traditional view with a single type $\eee$ for entities has another related drawback. It is unable to relate related predicates, although a usual dictionary does. 
A common noun like \ma{book} is usually viewed as a unary predicate "\underline{book}:$\eee\fl\ttt$"
while  a transitive verb like  \ma{read} is viewed as a binary predicate "\underline{read}:$\eee\fl\eee\fl\ttt$"
This gives the proper argument structure  of \emph{Mary reads a book.}  as ($\exists x:\eee book(x)\ and\ reads(Mary, x)$) but this traditional setting cannot relate  the predicates $\underline{book}$ and $\underline{read}$ --- while any dictionary does. If we had several types, as we shall do later on,  we could stipulate that the object of \ma{read} ought to be something that can be \ma{read},
that one can  \ma{read} and \ma{write} a \ma{book}.  Such connections like predicates like  \ma{book}, \ma{write}, \ma{read} would allow to interpret sentences like \ma{I finished my book} which usually means  \ma{I finished to  read my book} and sometimes \ma{I finished to write my book}.

Hence we need a more sophisticated type theory than the one initially used by Montague to filter semantically invalid sentences. But in some cases some flexibility is needed to accept and analyse sentences in which a word type is coerced into another type.  In sentence \ref{five}, in the context of a football match, the noun \ma{five} can be considered as a player i.e. a \ma{person} who plays the match with the number 5 jersey, who can \ma{run} and be \ma{fast}. 

There is a large literature on such lexical meaning transfers and coercions, starting from 1980 \cite{bierwisch1979,bierwisch1983,cruse1986lexical,nunberg-transfer} --- see also \cite{Lauer2004,blutner-pragma} 
for a more recent account of some theories.  
In those pioneering studies, the objective  is mainly to classify these phenomena, to find the rules that govern them. 
The quest of a computational formalisation that can be incorporated into an automated semantic analyser
appears with Pustejovsky's generative lexicon in 1991 \cite{Pustejovsky91,Pus95}. The integration of lexical issue 
into compositional semantics à la Montague and type theories appears with the work by Nicholas Asher \cite{AP01,asher-typedriven} which lead to the book \cite{Asher2011wow}, and  differently in some works of Robin Cooper with an intensive use of records from type theory to recover frame semantics with features and attributes inside type-theoretical  compositional semantics \cite{cooper-codygeqlic,Cooper2011lacl}

\subsection{Type theories for integrating lexical semantics} 

As the afore mentioned contribution suggest, finer-grained type theories are quite a natural framework both for formal semantics \`a la Montague and for selectional restriction and  coercions. Such a model must extend the usual ones into two directions: 
\begin{enumerate}  
\item \label{ttlex}
Montague's original type system and metalogic should be enriched to encompass lexical issues (selectional restriction and coercions), and 
\item \label{ttsem} 
the usual phenomena studied by formal semantics (quantifiers, plurals, generics) should be extended to this richer type system and so far 
only  Cooper and us did so
\cite{cooper-codygeqlic,Cooper2011lacl,
ChatLuo2012cslp,MootRetore2011coconat,MMR2013lenls,LMRS2012taln,Retore2012rlv} 
\end{enumerate} 

At the end of this paper, we shall provide a comparison of the current approaches, 
which mainly focus on \ref{ttlex}. Let us list right now what  the current approaches are: 

\begin{itemize} 
\item The system work with type based coercions and relies on some \emph{Modern Type Theory (MTT)}
\footnote{This name \emph{Modern Type Theory (MTT)} covers several variants of modern type theories, including 
Martin-L\"of type theory, the Predicative Calculus of (Co)Inductive Constructions (pCic), 
the Unifying Theory of dependent Types (UTT),... --- this later one being the closest to the system used by Zhaohui Luo} 
 --- this correspond to the work of  Zhaohui Luo \cite{Luo2011lacl,Luo2012lacl,LuoXue2012lacl,ChatLuo2012cslp}
\item The system work with type based coercions and relies on  usual typed 
$\lambda$-calculus extended with some categorical logic rules 
--- this approach by Asher \cite{AP01,asher-typedriven}  culminated  in his book \cite{Asher2011wow} 
\item The system work with term based coercions and relies on second order $\lambda$-calculus --- this is our approach, first introduced with Bassac, Mery, and further developed with Mery, 
Moot, Pr\'evot, Real-Coelho. \cite{BMRjolli,MPR2011taln,MPR2011cid,MootRetore2011coconat,MMR2013lenls,LMR2012cmlf,LMRS2012taln,Retore2012rlv,RealCoelhoRetore2013unilog,RealRetore2013jolli}
\end{itemize}

In fact our approach differs from the concurrent ones mainly because of the organisation 
of the lexicon and of the respective r{\^o}les of types and terms. 
Our approach can be said to be word driven, as it account for the (numerous) idiosyncrasies of natural language in particular  the different behaviour of words of the same type is coded by assigning them different terms, while others derive everything from the types. 

The precise type system we use, namely system \systF, does not make a big difference with other type theories, and as far as the presentation of the system is concerned, it is the \emph{simplest} of all systems, because it only contains four term building operations (two of them being the standard $\lambda$-calculus rules, the two other one being their second order counter part) and two reduction rules (one of them being the usual beta reduction and the other one being its second order counterpart). Dependent types, that types defined from terms are not avoided.

\section[A Montagovian generative lexicon for compositional semantic and lexical pragmatics]{A Montagovian generative lexicon\\ for compositional semantic and lexical pragmatics} 
\label{ltyn} 

We are to present our solution for introducing some lexical issues in a compositional framework \`a la Montague. 

\subsection{Guidelines for a semantic lexicon} 

We should keep in mind that whatever the precise solution presented, the following questions 
must be addressed in order to obtain a computational model, so here are the guidelines of our model: 

\begin{itemize} 
\item What is the logic for semantic representation?\\ \it 
We use \emph{many-sorted} higher order predicate calculus. As usual, the higher order can be reified in first order logic, so it can be first order logic, but in any case the logic has to be many sorted.  Asher \cite{Asher2011wow} is quite similar on this point, while Luo use Type Theory \cite{Luo2012lacl}. \normalfont 
\item What are the sorts?\\ \it The sorts are the base types. As discussed later on these sorts may vary from a small set of ontological kinds to any formula of one variable. We recently proposed that they correspond to \emph{classifiers} in language with classifiers: this give sorts a linguistically and cognitively motivated basis. \cite{MeryRetore2013nlpcs}
\normalfont 
\item What is the metalogic (glue logic) for meaning assembly?\\ \it 
We use second order $\lambda$-calculus (Girard system \systF) in order to factor operations that apply uniformly to family  types. For specific coercions, like ontological inclusion we use subtyping introduced in the present paper. 
Asher \cite{Asher2011wow} use simply typed $\lambda$-terms with additional categorical rules, while Luo also use Type Theory with coercive subtyping \cite{Luo2012lacl}. \normalfont 
\item What kind of information is associated with a word in the lexicon? \\ \it 
Here it will be a finite set of $\lambda$-terms, one of them being called the principal $\lambda$-term while the other are called optional. 
Other approaches make use more specific terms and rules. 
\normalfont 
\item How does one compose words and constituents for a compositional semantics?\\ \it 
We simply apply one $\lambda$-term to the other, following the syntactic analysis,  perform some transformations corresponding to coercions and presupposition, 
and reduce the compound by $\beta$-reduction. \normalfont 
\item How is rendered the semantic incompatibility of two components?\\ \it 
By type mismatch, between a function of type $A\fl X$ and an argument of type $B\neq A$, and others do the same. \normalfont 
\item How does one allow an \textsl{a priori} impossible composition?\\ \it 
By using the optional $\lambda$-terms, which change the types of at least one of the two terms being composed, the function and argument. Both the function and the argument may provide some optional lambda terms. 
Other approaches rather use type driven rules. \normalfont 
\item How does one allow and block felicitous and infelicitous copredications on various aspect of a word?\\ \it 
An aspect car be explicitly declared as incompatible with any other aspect. More recently we saw that linear types (linear system \systF) can account for  compatibility between arbitrary subsets of the possible aspects. \cite{MR2013nlcs}
\end{itemize}

Each word in the lexicon is given a principal term, as well as a finite number,  possibly nought,  optional terms that  licence type change and implement coercions. They may be inferred from an ordinary dictionary, electronic or not. Terms combine almost as usual except that there might be type clashes,  which accounts for infringement of   selectional restriction: in this case optional terms may be use to solve the type mismatch. In case they lead to different results these results should be considered as different possible readings --- just as the different readings with different quantifier scopes are considered by formal semantics as different possible readings of a sentence.

Let us first present the type and terms and thereafter we shall come back to the 
the composition modes. 

\subsection{Remarks on the type system for semantics}

We use a type system that resembles Muskens $Ty_n$ \cite{Muskens91} where the usual type of individuals, $\eee$ is replaced with a finite but  large set of base types $\eee_1,\ldots,\eee_n$ for individuals, for instance \emph{objects}, \emph{concepts}, \emph{events},... These base types  are the sorts of the many sorted logic whose formulae express semantic representations. The set of base types as well as their interrelations  can express some ontological relations as Ben Avi and Francez thought ten years ago 
\cite{BenAviFrancez2004cg}

 For instance, assume we have a many sorted logic with a sort 
$\zeta$ for animals, a sort $\phi$ for physical objects and a predicate $eat$ whose arguments are of respective sort $\phi$ and $\zeta$ 
the many sorted formula $\forall z:\zeta\ \exists x:\phi\ eat(z,x)$ is rendered in type theory by the $\lambda$-term: 
$\forall^\zeta (\lambda z^\zeta \exists^\phi \lambda x^\phi ((eat\ x) z)$ with $eat$ a constant of type $\phi \fl \zeta \fl \ttt$. Observe that the type theoretic formulation requires a quantifier for each sort $\alpha$ 
of object, that is a constant $\forall^{\alpha}$ of type $(\alpha\fl \ttt) \fl\ttt$. 
\footnote{We do not speak about interpretations, but if one wishes to, we do not necessarily ask for the usual requirement that sorts are disjoint: this is coherent with the fact that in type theory, nothing prevents a pure term to have several types.} 

What are the base types? We have a tentative answer, but we cannot be too 
sure of the answer. 
Indeed, this is a subtle question depending on ones philosophical convictions, and also of the expected precision of the semantic representations.\footnote{For instance, a dictionary says that pregnant can be said of a \ma{woman or  female animal}, but can it be said of a \ma{grandma} or of a \ma{veal}?} 
but it does not really interfere with the formal and computational model we present here. 
Let us mention some natural sets of bases types are, from the smallest to the largest: 
\begin{enumerate} 
\item \label{Montague} \emph{A single type $\eee$ for all entities} (but as seen above it cannot account for lexical semantics) 
\item \label{sorts} \emph{A very simple ontology} distinguishing events, physical objects, living entities, concepts, ... (this resembles Asher's position) 
\item \label{classifiers} Many Asian languages (Chinese, Japanese, Korean, Malay, Burmese, Nepali,...) and all Sign Languages, have \emph{classifiers} that are pronouns specific to classes of nouns (100--400) especially detailed for physical objects that are handled, animals.\emph{There are almost no classifiers in European languages. Nevertheless a word like \ma{head} in \ma{Three heads of cattle.} can be considered as a classifier.} Hence classifiers are a rather natural set of base types, or the importation of the classifiers of a language  in one that does not have any. But we do not claim that this is the definitive answer. For instance, for a specific task, some other set of base types may be better. 
 \cite{MeryRetore2013nlpcs}
\item \label{cn} \emph{A type per common noun} as proposed by Luo in \cite{Luo2012lacl}) 
\item \label{form} \emph{A type for every formula} with a single free variable as suggested by some colleague (N. Asher or F. Corblin) after a talk of mine.  
\end{enumerate} 

Our opinion is that types should be cognitively natural classes and rich enough to express selectional restrictions. 
Whatever types are, there is a relation between types and  properties. With base types as in \ref{form}, the correspondence seems quite clear, 
but,  because types can be used to express new many sorted formulae, the set of types is in this case defined as a least fixed point. 
For other sets of base types, e.g. \ref{cn} or \ref{sorts} for each type $\tau$ there should be a corresponding predicate which recognises $\tau$ entities 
among entities of a larger type. For instance, if there is a type \emph{dog} there should be a predicate $\pred{dog}:\alpha\fl\ttt$
but what should be $\alpha$ the type of its argument? Should it be \ma{animal}, \ma{animate},... the simplest solution is to assume a type of all individuals, that is Montague's $\eee$, and to say that corresponding to any base type $\tau$, there is a predicate, namely $\pred{\tau}$ of type $\eee\fl\ttt$\footnote{An alternative solution, used by us and others 
\cite{Retore2012rlv,ChatLuo2012FG}
 would be $\Pi \alpha.\ \alpha\fl\ttt$, using quantification over types to be defined in next section.} 

Let us say here a remark on the predicate constants in the language. 
If a predicate constant, say $Q$ is given with type $u\fl \ttt$ with $u\neq \eee$ which sometimes is more natural  there is an obvious extension $Q_e$ which should be interpreted as false for any object that cannot be viewed as an $u$-object. Given predicate in the language do also have restrictions, $Q|_v$ 
which is defined as $Q$ on $q\cap v$ where $q$ is  the domain of $Q$ and false elsewhere.

\subsection[Many sorted formulae in second order lambda calculus]{\ltyn: many sorted formulae in second order lambda calculus} 

Since we have many base types, and many compound types as well, it is quite convenient and almost necessary 
to define operations over family of similar terms with different types,  to have some flexibility in the typing, and to have terms that act upon families of terms and types. 
Hence we shall extend further  $Ty_n$  into $\ltyn$ by using Girard's system \systF as the type system 
\cite{Girard2011blindspot,Girard71}. System \systF involves quantified types whose terms can be specialised to any type. 

The types of \ltyn are defined as follows: 
 \begin{itemize} 
\item 
Constants types $\eee_i$ and $\ttt$, as well as type variables $\alpha,\beta,\ldots$ are types. 
\item 
Whenever $T$ and $\alpha$ respectively are a type and a type variable $\Pi\alpha.\ T$ is a type. The type variable may or may not occur in the type $T$. 
\item 
Whenever $T_1$ and $T_2$ are types, $T_1\fl T_2$ is a type as well.
\end{itemize}

The terms of $\ltyn$, which encode proofs of quantified propositional intuitionistic logic, are defined as follows: 
\begin{itemize} 
\item A variable  of type $T$ i.e. $x:T$ or  $x^{T}$  is a \emph{term}, and there are countably many variables of each type.
\item In each type, there can be a countable set of constants of this type, and a constant of type $T$ is a term of type $T$. Such constants are needed for logical operations and for the logical language (predicates, individuals, etc.). 
\item 
$(f\ \tau)$ is a term of type $U$ whenever $\tau:T$ and  $f:T\fl U$. 
\item 
$\lambda x^{T}.\ \tau$ is a term of type $T\fl U$ whenever $x:T$, 
and $\tau:U$.  
\item $\tau \{U\}$ is a term of type $T[U/\alpha]$
whenever $\tau:\Lambda \alpha.\ T$, and $U$ is a type. 
\item $\Lambda \alpha. \tau$ is a term of type $\Pi \alpha. T$
whenever $\alpha$ is a type variable, and  $\tau:T$ a term without any free occurrence of the type variable $\alpha$ in the type of a free variable of $\tau$.  
\end{itemize}

The later restriction is the usual one on the proof rule for quantification in propositional logic: one should not conclude that $F[p]$ holds for any  proposition $p$
when assuming $G[p]$ --- i.e. having a free hypothesis of type $G[p]$. 

The reduction of the terms in system F or its specialised version \ltyn is defined by the two following reduction schemes that resembles each other:  
\begin{itemize} 
\item $(\lambda x. \tau) u$ reduces to $\tau[u/x]$ (usual $\beta$ reduction). 
\item $(\Lambda \alpha. \tau) \{U\}$  reduces to $\tau[U/\alpha]$ (remember that $\alpha$ and $U$ are types). 
\end{itemize} 

As an example, we earlier said that in $Ty_n$ we needed a first order quantifier per sort (.e. per base type). 
In \ltyn it is sufficient to have a single quantifier $\forall$, that is a constant of type $\Pi \alpha.\ (\alpha \fl \ttt)\fl \ttt$ . Indeed, this quantifier can be specialised to specific types, for instance to the base type $\zeta$, yielding $\forall \{\zeta\}: (\zeta \fl \ttt)\fl \ttt$, or even to properties  of $\zeta$ objects, which are of type $\zeta\fl\ttt$, yielding  
$\forall \{\zeta\fl\ttt\}: ((\zeta \fl \ttt) \fl \ttt)\fl \ttt$. We actually do quantify over higher types, for instance in the examples below respectively quantify over propositions with a human subject, and the next one over propositions: 

\begin{exe}
\ex 
He did everything he could to stop them. 
\ex 
And he believes whatever is politically correct and sounds good.
\end{exe} 

As Girard showed \cite{Girard2011blindspot,Girard71}
reduction is strongly normalising and confluent 
\textit{every term of every type admits a unique normal form which is reached no matter how one proceeds.} 
\footnote{This is one way to be convinced of the soundness of \systF, which defines types depending on other types including themselves: as it is easily observed that there are no normal closed terms of type $\Pi X.\ X\equiv \bot$ the system is necessarily coherent. Another way is to construct a concrete model, called coherence spaces, where types are interpreted as countable sets with a  binary relation (coherence spaces), and terms up to normalisation are interpreted as structure preserving functions (stable functions). \cite{Girard2011blindspot}} The normal forms (which can be asked to be $\eta$-long) can be characterised as follows (for a reference see e.g.  \cite{Huet76}) : 

\begin{proposition}\label{head} 
A normal 
$\Lambda$-term $\mathcal{N}$ of system \systF, $\beta$ normal and $\eta$ long to be precise, has the following structure: 

\begin{center} 
\begin{tabular}{crcl} 
& & & sequence\ of\ $\{\cdots\}$ and $(\cdots)$\\ 
&sequence\ of\ & head\ & applications\ to\ types\ $W_k$ \\ 
&$\lambda$ and\ $\Lambda$ abstractions  &
variable &
and\ normal\ terms\ $t_l^{X_l}$\\ 
$\mathcal{N}\ =$&$\overbrace{(\ \lambda x_i^{X_i}\ |\ \Lambda X_j\ )^*}$ & 
$\overbrace{h^{(\Pi X_k |X_l\fl)^* Z}}$ &
$\overbrace{(\ \{W_k\}\ |\ t_l^{X_l})^*\ )}$ 
\end{tabular}
\end{center} 

\end{proposition} 

This has  a good consequence for us, see e.g. \cite[Chapter 3]{MootRetore2012lcg}: 

\begin{prop}[\ltyn terms as formulae of a many-sorted logic] If the predicates, the constants and the logical connectives and quantifiers 
are the ones  from a many sorted logic of order $n$ (possibly $n=\omega$) then the normal terms of $\ltyn$ of type $\ttt$ unambiguously correspond to many sorted  formulae of order $n$. 
\end{prop}

Let us  illustrate how \systF factors uniform behaviours. 
Given types $\alpha$, $\beta$, two predicates $P^{\alpha\fl \ttt}$, $Q^{\beta\fl \ttt}$,  over entities of respective kinds $\alpha$ and  $\beta$ for any $\xi$ with two morphisms from  $\xi$ to $\alpha$ and to $\beta$, see figure \ref{polyandfig} 
\systF contains a term that can coordinate the properties $P,Q$ of (the two images of) an entity of type $\xi$,
every time we are in a situation to do so --- i.e. when the lexicon provides the morphisms.  

\begin{term} \label{polyandterm} 
[Polymorphic AND]  is defined as $\Land$ =\newline 
$\Lambda \alpha \Lambda \beta
\lambda P^{\alpha \fl \ttt} \lambda Q^{\beta\fl \ttt} 
 \Lambda \xi \lambda x^\xi 
 \lambda f^{\xi\fl\alpha} \lambda g^{\xi\fl\beta}.\ 
(\textrm{and}^{\ttt\fl\ttt\fl\ttt} \ (P \ (f \ x)) (Q \ (g \  x))) 
$ 
\end{term} 

\begin{figure}
\begin{center}
\includegraphics[scale=0.3]{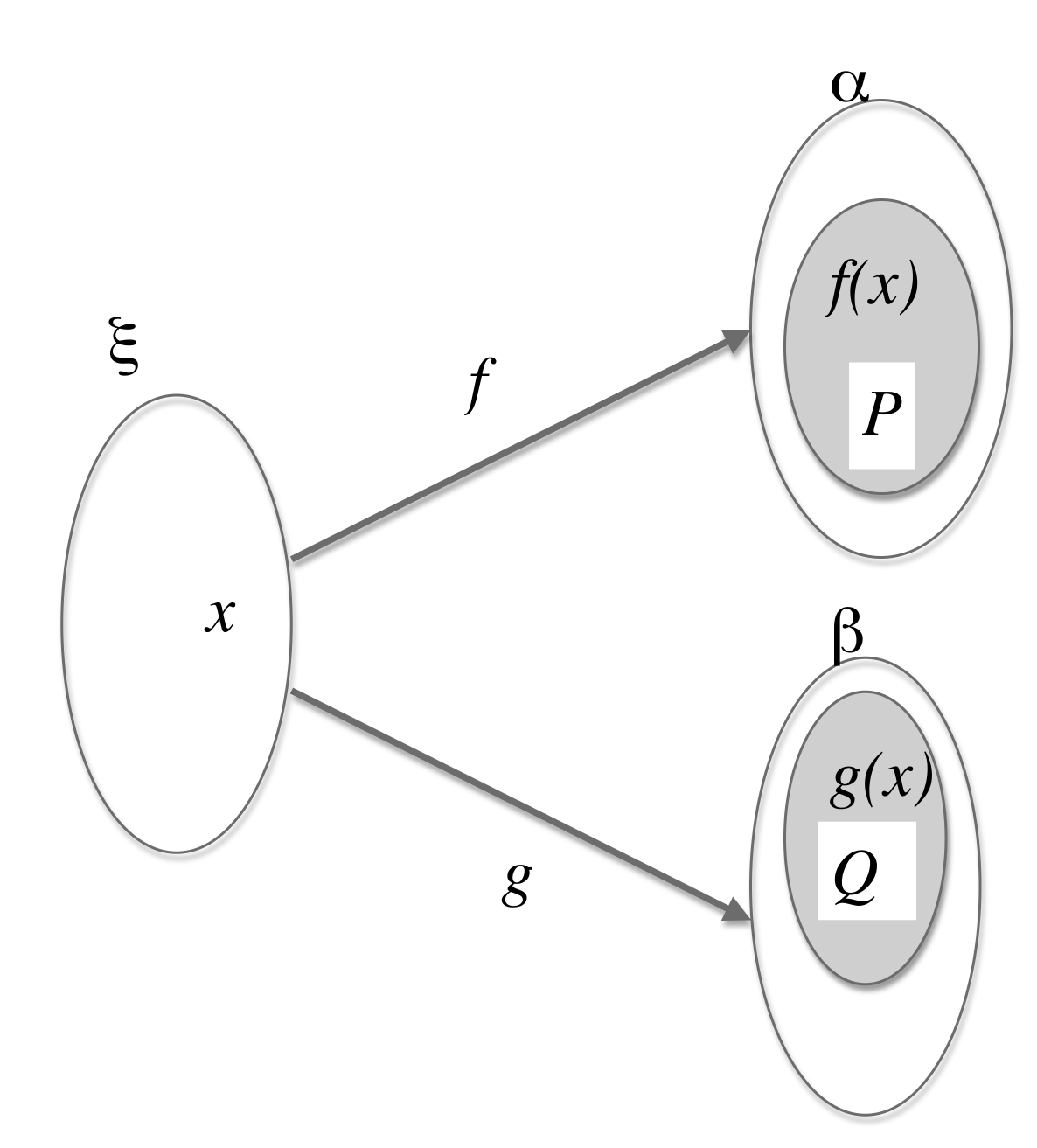} 
\end{center} 
\caption{Polymorphic conjunction: $P(f(x))\& Q(g(x))$ with $x:\xi$, $f:\xi\fl\alpha$, $g:\xi\fl\beta$.}  
\label{polyandfig}
\end{figure}

This can apply to say, a \ma{book},   that can be \ma{heavy}  
as a \ma{physical object}, and \ma{interesting} as an \ma{informational content} 
--- the limitation of possible over generation is handled by the \emph{rigid} use of possible transformations, including identity to be defined thereafter. 

\subsection{Organisation of the lexicon and rules for meaning assembly}
\label{orga} 

The lexicon associate each word $w$ with a \emph{principal $\lambda$-term} $[w]$ which basically is the Montague term reminded earlier, except that the types appearing in $[w]$ belong to a much richer typed system. In particular, the numerous base types can impose some selectional restriction. In addition to this principal term, there can be \emph{optional $\lambda$-terms} also called \emph{modifiers} or \emph{transformations} to allow, in some cases, composition that were initially ruled out by selectional restriction. 

There are two ways to solve a type conflict using those modifiers. \emph{Flexible modifiers} can be used without any restriction. \emph{Rigid modifiers} turn the type, or the sense of a word,  into another one which is \emph{incompatible}  with other types or senses. For a technical reason, the identity which is always a licit modifier is also specified to be flexible or rigid. In this later rigid case, it means that the original sense is incompatible with any other sense, although two other senses may be compatible. 
Consequently, every modifier, i.e. optional $\lambda$-term is declared, in the lexicon, to be either a rigid modifier, noted  $\rig$ or a flexible one, noted $\flex$. More subtle compatibility relations between senses can be represented by using the linear version of system \systF as we did in \cite{MR2013nlcs}

\begin{figure} 
$$ 
\begin{array}{l|l|rl} 
\mbox{word} & \mbox{principal\ $\lambda$-term} & \multicolumn{1}{l}{\mbox{optional\ $\lambda$-terms}} & \mbox{rigid/flexible}\\ \hline 
book & \pred{B}:\eee\fl\ttt & Id_B:B\fl B &\flex \\ 
& & b_1:B\fl \phi &\flex \\ 
& & b_2:B\fl I & \flex\\ 
\hline 
town & \pred{T}:\eee\fl\ttt & Id_T:T\fl T &\flex \\ 
& & t_1:T\fl F&\rig \\ 
& & t_2:T\fl P &\flex\\ 
& & t_3:T\fl Pl &\flex\\ 
\hline 
Liverpool & liverpool^T & Id_T:T\fl T &\flex \\ 
& & t_1:T\fl F &\rig \\ 
& & t_2:T\fl P &\flex \\ 
& & t_3:T\fl Pl &\flex\\ 
\hline 
vast & vast:Pl\fl\ttt & \\ 
\hline 
voted & voted:P\fl\ttt& \\
\hline 
won & won:F\fl\ttt&\\  
\end{array}
$$
where the base types are defined as follows: 
\begin{center}
\begin{tabular}{ll}
$\phi$ & physical objects\\  
$B$ & book \\ 
\end{tabular} 
\quad 
\begin{tabular}{ll}
$I$ & information\\ 
$T$ & town \\ 
\end{tabular} 
\quad 
\begin{tabular}{ll}
$P$ & people \\ 
$Pl$ & place \\
\end{tabular} 
\end{center} 
\caption{A sample lexicon}
\label{lexicon} 
\end{figure}

The reader may be surprised that we repeat the morphisms in the lexical entries, rather than having general rules.  For instance, one could also consider morphisms that are not anchored in a particular entry: in particular, they could implement the ontology at work in \cite{Pus95} as the type-driven approach of Asher does \cite{Asher2011wow}. 
For instance, a place (type $Pl$) could be viewed as a physical object (type $\phi$) with a general morphism $P2\phi$ turning places into physical objects that can be \ma{vast}. We are not fully enthusiastic about a general use of such rules since it is hard to tell whether they are flexible or rigid. As they can be composed they might lead to incorrect copredications, while their repetition inside each entry offers a better control of incorrect and correct copredications. One can think that some meaning transfer differs although the words have the same type. 
An example of such a situation in French is provided the words \ma{classe} and \ma{promotion}, which both refer to groups of pupils. 
The first word  \ma{classe}  (English: \ma{class}) can be coerced into the room where the pupils are taught, (the \ma{classroom}), while the second, \ma{promotion}  (English: \ma{class} or \ma{promotion}) cannot. 

There nevertheless exist ontological inclusions that are better represented by rules on types, like \ma{car} that are \ma{vehicles} that are \ma{artefacts}. This is the reason why we also allow for optional terms that are available for all words of the same type. This is done by \emph{subtyping} and more precisely by the notion of coercive subtlyping that is introduced in section \ref{subtyping}.

\section{A proper account of meaning transfers} 

In this section we shall see that the lexicon we propose, provides a proper account of the lexical phenomena that motivated its definition:  ill typed readings are rejected, coerced readings are handled, felicitous copretication are analysed while infelicitous ones are rejected. Some particular case of coerced readings are given a finer analysis as the polysemy of deverbals (nouns derived verbs, like \ma{construction}), or fictive motion. Finally we introduce coercive subtyping for system \systF which handles general coercions corresponding to ontological inclusion. 

\subsection{Coercions and copredication} 

One can foresee what is going to happen, using the lexicon given in figure \ref{lexicon} with sentences like: 
\begin{exe}
\ex \label{ll} 
Liverpool is vast. 
\ex \label{llv} 
Liverpool is vast and voted (last Sunday). 
\ex \label{lvw}    \#
Liverpool voted and won (last Sunday). 
\end{exe}

Our purpose is not discuss whether this or that sentence is correct, nor whether this or that copredication is felicitous, but to provide a formal and computational model which given sentences that are assumed to be correct, derives the correct readings, and which given sentences that are said to be incorrect, fails to provide a reading. 


\begin{itemize}
\item[Ex. \ref{ll}]
This sentence leads to a type mismatch $vast^{Pl\fl\ttt}(Liverpool^T))$, since \ma{vast} applies to \ma{places} (type $Pl$) and not to \ma{towns} as \ma{Liverpool}. 
It  is solved using the optional term $t_3^{T\fl Pl}$ provided by the entry for \ma{Liverpool}, which turns a town ($T$) into a place ($Pl$) 
$vast^{Pl\fl\ttt}(t_3^{T\fl Pl} Liverpool^T))$ --- a single optional term is used, the \flex / \rig difference is useless. 
\item[Ex. \ref{llv}] 
In the second example, the fact that Liverpool is vast is derived as previously, and the fact Liverpool voted is obtained from the transformation of the town into people, that can vote. The two can be conjoined by the polymorphic \ma{and} defined above 
as  term \ref{polyandterm} ($\Land$) 
because these transformations are flexible: one can use one and the other. 
We can make this precise using only the rules of the type calculus. 
The syntax yields the predicate $(\Land (is\_vast)^{Pl\fl \ttt} (voted)^{P\fl\ttt})$ and consequently 
the type variables should be instantiated by $\alpha:=Pl$ and $\beta:=P$ and the exact term is 
$\Land \{Pl\} \{P\} (is\_vast)^{Pl\fl \ttt} (voted)^{P\fl\ttt}$ which reduces to:  
$ \Lambda \xi \lambda x^\xi  \ 
 \lambda f^{\xi\fl\alpha} \lambda g^{\xi\fl\beta}  
(\textrm{and}^{\ttt\fl\ttt)\fl\ttt} \ (is\_vast \ (f \ x)) (voted \ (g \  x)))$. 

Syntax also says this term is applied to \ma{Liverpool}. 
which forces the instantiation $\xi:=T$ and the term corresponding to the sentence is after some reduction steps,\\  
$ \lambda f^{T\fl Pl} \lambda g^{T\fl P}  
(\textrm{and} \ (is\_vast \ (f \ Liverpool^T)) (voted \ (g \  Liverpool^T))))$. Fortunately the optional $\lambda$-terms 
$t_2:T\fl P$ and  $t_3:T\fl Pl$ are provided by the lexicon, and they can both be used, since none of them is rigid.
Thus we obtain, as expected\\  
$(\textrm{and} \ (is\_vast{Pl\fl\ttt} \ (t_3^{T\fl Pl} \ Liverpool^T)) (voted^{Pl\fl\ttt} \ (t_2^{T\fl P} \  Liverpool^T)))$ 
\item[Ex. \ref{lvw}]
The third example is rejected as expected. Indeed, the transformation of the town into a football club prevents any other transformation (even the identity) to be used in the polymorphic and that we defined above. We obtain the same term as above, with $won$ instead of $is\_vast$. 
The term is: $ \lambda f^{T\fl Pl} \lambda g^{T\fl P}  
(\textrm{and} \ (won \ (f \ Liverpool^T)) (voted \ (g \  Liverpool^T))))$ 
and the lexicon provides the two morphisms that would solve the type conflict, but one of them is rigid, i.e. we can solely use  this one. Consequently the sentence is semantically invalid. 
\end{itemize} 

\subsection{Fictive motion}
A rather innovative extension is to apply this technique to what Talmy called \emph{fictive motion} \cite{talmy99fictive}. Under certain circumstances, a path may introduce a virtual traveller following the path, as in sentences like: 
\begin{exe}
\ex Path GR3 descends for two hours.
\end{exe} 
Because of the duration, one cannot consider that the vertical coordinate decreases as the curvilinear abscissa increases.   
One ought to consider someone who follows the road. We model this by one morphism associated with the \ma{Path GR3} and one with \ma{descends}. The  first coercion turns the \ma{Path GR3} from an immobile object into an object of type  \ma{path} that can be followed and the second one coerce \ma{descends} into a verb that acts upon a \ma{path} object and introduce an individual following the path downwards --- this individual, which does not need to exist,  is  quantified, yielding a proposition that can be paraphrased as \ma{any individual following the path goes downwards for two hours}. 
\cite{MPR2011taln,MPR2011cid}

\subsection{Deverbals} 

Deverbals are nouns that correspond to action verbs, as \ma{building} or \ma{signature}. Usually they are ambiguous between result and process. We showed that our idiosyncratic model is well adapted since their possible senses vary from one deverbal to another, even if the verbs are similar and the suffix is the same. 

\begin{exe}
\ex The building took three months. 
\ex The building was painted white. 
\ex * The building that took three months was painted white. 
\ex The signature was illegible. 
\ex The signature took three months. 
\ex * Although it took three months the signature was illegible. 
\ex Although it took one minute, the signature was illegible. 
\end{exe}

We showed that a systematical treatment of deverbal meaning as the one proposed by the type-driven approach does not properly account for the data. Indeed, the possible meanings of a deverbal are more diverse than result and event, and there are no known rules to make sure the deverbal refers to the event. Consequently, words must include in the lexical  informations such at the possible meanings of the deverbal. These meanings can be derived from the event expressed by the verb, they usually include the event itself (but not always),  the result (but not always), 
and other meanings as well like the place where the event happens (e.g. English noun \ma{pasture}). 
This lexical information can be encoded in our framework, with one principal meaning and optional terms for accessing other senses and the flexibility or rigidity of these optional terms --- they are usually ridig, and copredication on the different senses of a deverbal is generally infelicitous. 
W successfully applied our framework and treatment to the semantic of deverbals to the  of restrictions of selection (both for the deverbal and for the predicate that may apply to the deverbal) to meaning transfers,  and to the  felicity  of  copredications on different senses of a deverbal. 
\cite{RealCoelhoRetore2013unilog,RealRetore2013jolli}

\subsection{Coercive subtyping and ontological inclusions}
\label{subtyping} 

As we said earlier on, ontological inclusions like \ma{Human beings are animals.},  would be better modelled by optional terms that are available for any word of the type, instead of anchoring them in words and repeating these terms for every word of this type. The model we described can take these subtyping inclusions into account as standard coercions, by specifying that a word like \ma{human being} introduces a transformation into an \ma{animal}. But this is somehow heavy, since one should also say that \ma{human beings} are  \ma{living beings} etc. 
Any predicate,  that applies to a class, also applies to an ontologically smaller class. For instance, 
 \ma{run} that applies to \ma{animals} also applies to  \ma{human beings}, because the \ma{human} is a subtype of \ma{animals}. These subtype coercions looks type driven, and, consequently, would be more faithfully modelled with a proper notion of subtyping.

Coercive subtyping,  introduced by Luo and Soloviev\cite{LuoSolovievXue2013ic,Soloviev00coercioncompletion}  for variants of Martin-L\"of type theory, 
corresponds quite well to these particular transformations. It starts with a transitive and acyclic set of coercions between base types, with at most  one coercion between any two base types, and ontological inclusions fulfil this condition. 
Indeed, such ontological inclusions when viewed as functions always are the identity on objects, hence there cannot be two different manners to map them in the larger type. 
Furthermore, other notions of subtyping that have been studied for higher order type theories are  very complicated with tricky restriction on the subtyping rules. \cite{Cardelli94Fsubtyping,LongoMilstedSoloviev2000lc}

Coercive subtyping, noted $A_0<A$, can be viewed as a short  hand for allowing a predicate or a function which applies to $A$-objects to apply to an argument whose type $A_0$ is not the expected type $A$ but a subtype $A_0$ of $A$. 
Hence coercive application is exactly what we were looking for: 

\begin{center}
coercive application \medskip 

\begin{prooftree} 
f:A\fl B \quad u:A_0 \quad A_0< A
\justifies 
(f\ a): B
\end{prooftree} 
\end{center}

The subtyping judgements, which have the structure of categorical combinators, are derived with very natural rules given in figure \ref{subtypingrules}. 
These rules simply encode transitivity, covariance and contravariance of implicative types (arrow types), and quantification over type variables. 

\begin{figure} 
\begin{center} 
\dotfill 

\medskip 

transitivity   \medskip 

\begin{prooftree} 
A<B\qquad B<C
\justifies 
A<C 
\end{prooftree}

\medskip 

\dotfill 

\medskip

covariance and contravariance of implication  \medskip 

\begin{prooftree} 
A<B\qquad C<D
\justifies 
D\fl A < C\fl B
\end{prooftree} 
\hfill
\begin{prooftree} 
A<B
\justifies 
T\fl A < T\fl B
\end{prooftree} 
\hfill 
\begin{prooftree} 
A<B
\justifies 
B\fl T <  A \fl T 
\end{prooftree}  

\medskip 

\dotfill 

\medskip 

quantification over types \medskip

\begin{prooftree} 
U<T[X]
\justifies 
U < \Pi X. T[X]
\using X\ not\ free\ in\ U
\end{prooftree} 
\hfill 
\begin{prooftree} 
U < \Pi X. T[X]
\justifies 
U<T[W]
\end{prooftree} 
\medskip 

\dotfill

\end{center}
\caption{Rules for coercive subtyping in system \systF}  
\label{subtypingrules}
\end{figure}

It should be observed that, given constants 
$\cij$
representing the coercions from $\eee_i$ to $\eee_j$, 
any coercion derivable coercion $T<U$ can be depicted by a linear $\Lambda$-term $m:U$  
of system $F$ or \ltyn 
with a single occurrence of the free variable $x:T$ and occurrences of the constants $\cij$. 
The construction of the term according to the derivation rules is defined as follows:

\begin{itemize} 
\item 
transitivity   \medskip 

\begin{prooftree} 
x:A< t:B\qquad y:B< u:C
\justifies 
x:A <  u[y:=t]:C
\end{prooftree} 

\bigskip 

\item 
covariance and contravariance of implication  \medskip 
\begin{itemize} 
\item 
\begin{prooftree} 
x:A< t:B\qquad z:C< u:D
\justifies 
f:D\fl A < \lambda z^C t[x:=f(u)]: C\fl B
\end{prooftree} 

\bigskip

\item 
\begin{prooftree} 
x:A<t:B
\justifies 
f:T\fl A < \lambda w^T t[x:=f(w)]:T\fl B
\end{prooftree} 

\bigskip

\item 
\begin{prooftree} 
x:A<t:B
\justifies 
g:B\fl T <  \lambda x^A. g(t): A \fl T 
\end{prooftree}  

\bigskip

\end{itemize} 
\item 
quantification over types \medskip 
\begin{itemize} 
\item 
\begin{prooftree} 
u:U<t:T[X]
\justifies 
u:U < \Lambda X. t:\Pi X. T[X]
\using X\ not\ free\ in\ U
\end{prooftree} 

\bigskip

\item 
\begin{prooftree} 
u:U < t:\Pi X. T[X]
\justifies 
u:U<t\{W\}:T[W]
\end{prooftree} 

\bigskip

\end{itemize}
\end{itemize} 

As an easy induction shows that: 

\begin{proposition} 
All terms derived in this system are linear, with a single occurrence of a single free variable (whose type is on the left of \ma{<}). 
\end{proposition} 

From this one easily concludes that: 

\begin{proposition} 
Not all $\Lambda$-terms of system \systF can be derived in the subtyping   system. 
\end{proposition}

Any derivation $c$ of $\eee_i < \eee_j$ is equivalent to a coercion $\cij$, i.e. our derivation system does not introduce new coercions between atomic types. This kind of result is similar to coherence in categories: 
given a compositional graph  $G$, the free cartesian categories over $G$ does not contain any extra morphism between object from the compositional graph. 
Here is the precise formulation of this coherence result: 

\begin{proposition} 
Given a $\eee_i<\eee_j$-derivation whose associated $\Lambda$-term is $\tilde{\mathcal{C}}$, 
the normal form $\mathcal{C}$ of $\tilde{\mathcal{C}}$ is a compound of $\cij$ applied to $x:\eee_i$, which, because of the assumptions on the $\cij$ is some $\chk$. 
\end{proposition}

\begin{proof} 
As seen above, a deduction of $T<U$ clearly corresponds to a linear $\Lambda$-terms of system \systF, whose only free variable is $x:T$ with the $\cij$ as constants. Hence it has a normal from which also 
has a single free variable is $x:T$ and the $\cij$ as constants.

Let us show that any normal $\Lambda$-term $\mathcal{C}$ of type $\eee_j$ with a single free variable $x:\eee_i$
and constants $\cij:\eee_i\fl\eee_j$ is a compound of $\cij$ applied to $x^{\eee_i}$, i.e. a term of $C_i$:
\begin{itemize} 
\item 
$x^{\eee_i}\in C_i$
\item 
if $c^{\eee_j}\in C_i$ then $(\cjk(c))^{\eee_k}\in C_i$
\end{itemize} 

We proceed by induction on the number of occurrences of variable and constants in the normal term $\mathcal{C}$, whose from is, as said in proposition \ref{head}: 

\begin{center} 
\begin{tabular}{crcl} 
&sequence\ of\ & head\ & sequence\ of\ $\{\cdots\}$ and $(\cdots)$\ applications\\ 
&$\lambda$ and\ $\Lambda$ abstractions  &
variable &
to\ types\ $W_k$ and\ normal\ terms\ $t_l^{X_l}$\\ 
$\mathcal{C}\ =$&$\overbrace{(\ \lambda x_i^{X_i}\ |\ \Lambda X_j\ )^*}$ & 
$\overbrace{h^{(\Pi X_k |X_l\fl)^* Z}}$ &
$\overbrace{(\ \{W_k\}\ |\ t_l^{X_l})^*\ )}$ 
\end{tabular}
\end{center} 

If the term $\mathcal{C}$ corresponds to a proof of $\eee_i<\eee_j$ 
there is no $(\ \lambda x_i^{X_i}\ |\ \Lambda X_j\ )$
in front, because the $\eee_j$  is neither of the form  $U\fl V$ nor of the form $\Pi X.\ T[X]$. 
What may be the head variable? It is either the only free variable of this term, namely $x^\eee_i$, or a constant i.e. a $\ckl$. 
\begin{itemize} 
\item 
If the head variable is $x^\eee_i$ then, because of its type, there is no application to a type or to a normal term $(\ \{W_k\}\ |\ t_l^{X_l})^*\ )$ arguments, hence $\eee_i=\eee_j$ and the normal form is $x^{\eee_i}$, which is in $C_i$ 
\item 
If the head variable is is some $\ckl$, which because of its type, 
may only be applied to a normal term $t_l^{X_l}$ of type $\eee_k$. 
This normal term is a normal term of type $\eee_k$ with $x_{\eee_i}$ as its single free variable and the constants $\cjl$. As $t_l^{X_l}$  has one symbol less than $\mathcal{C}$, 
we can conclude that $t_l^{X_l}$ is in $C_i$ hence $\mathcal{C}\in C_i$. 
\end{itemize} 

Hence in any case the normal form $\mathcal{C}:\eee_j$ 
of the term  $\tilde{\mathcal{C}}:\eee_j$ is in $C_i$. 

Now, given that the coercions $\cij$ enjoys $\ckj \circ \cij= \cik$ (as part of our condition on base coercions)  
it is easily seen  that the only term of type $\eee_j$ in $C_i$ is $\cij$. 
\end{proof} 
 
We think that this coherence result can be improved by showing that there is at most one normal term corresponding to a derivation $S<T$, although the proof is likely to use some variant of reducibility candidates.

\subparagraph{An alternative} 

The rules for coercive sub tying follow a natural deduction style, as lambda terms of system \systF. Nevertheless, an alternative formulation of the quantifier elimination rule which requires to have identity axioms (whose term is identity)  to derive obvious sub tying relations. 

\begin{center} 
alternative quantifier elimination rule (sequent calculus style) 
\medskip 

\begin{prooftree} 
s:S[T] < t:U
\justifies 
\dot s:\Pi X. S[X] <  t[s:=\dot s\{T\}]
\end{prooftree} 
\end{center} 

\section{Compositional semantics issues: determiners, quantifiers, plurals} 
\label{determiner} 

So far we focused on phenomena in \emph{lexical semantics} that are usually left out of standard models but properly 
mastered by our model. But we must also have a look at compositional semantics, that is a as the logical structure of a sentence, to see whether our model still properly analyses what standard compositional models do, and, possibly provide better analysis. 
Hopefully sentence structure are correctly analysed but furthermore our extended setting  is quite  appealing for 
some classical issues in \emph{formal semantics} like determiners and quantification, or  plurals, as we show in this section. 

\subsection{Determiners and quantifiers} 

The examples presented so far only involved proper names because the determiners and quantifiers are a bit more complex than in the usual montagovian setting, let us see how they work.  

In order to integrate lexical issues into compositional semantics which closely follows syntax, 
we should at least describe the behaviour of  determiners and quantifiers in our framework. 
We adopt the view of quantified, definite, and indefinite  noun phrases as \emph{individual terms} by using generic elements (or choice functions) as initiated by Russell and formalised by Hilbert, Ackerman and Bernays see e.g. \cite{HBvol2}
and adapted to linguistics by researchers like von Heusinger see e.g. \cite{EgliHeusinger1995,Heusinger1997,Heusinger2004}.

How do we adapt our model, in particular the typing, if instead of \ma{Liverpool} the examples used \ma{The town}, \ma{A town}, \ma{All towns},  or \ma{Most towns}? Indefinite determiners, quantifiers, generalised quantifiers,... usually are viewed as functions from two predicates to propositions, one expressing the restriction and the other the main predicate see e.g. 
\cite{PetersWesterstahl2006quantifiers}

As we said, and this is especially true in  a categorial setting as the one Moot implemented 
\cite{moot10grail}
the syntactic structure closely corresponds to the semantic structure. 
But the usual treatment of quantification that we saw in subsection \ref{Montague} infringe 
this correspondence: 

\begin{exe}
\ex \label{keith} 
\emph{sentence:} Keith played some Beatles song. 
\ex 
\emph{semantical structure:} (some (Beatles songs)) ($\lambda x$ Keith played $x$)
\ex 
\emph{syntactical structure:}  (Keith (played (some (Beatles\ song))))
\end{exe} 

Another criticism that applies to the usual treatment of quantifiers is the symmetry that it wrongly introduces between the main predicate and the class over which one quantifies. For instance, the two sentences below (\ref{epc},\ref{ecp}) usually have the same logical form (\ref{fepc}): 

\begin{exe}\label{politician} 
\ex \label{epc} Some politician are crooks.
\ex \label{ecp} ? Some crooks are politicians. 
\ex \label{fepc} $\exists x. politician(x) \& crook(x)$
\end{exe}

Hence, in accordance with syntax, we prefer to consider that a quantified noun phrase is by itself some individual --- a generic one which does not  refer 
to a precise individual nor to a collection of individuals. As \cite{Heusinger1997} we use a $\eta$ for indefinite determiners (whose interpretation picks up a new element)  and $\iota$ for definite noun phrases\footnote{Actually \cite{Heusinger1997} writes $\epsilon$ instead of $\iota$. We do not follow his notation  because we also use Hilbert's $\epsilon$ with its traditional meaning.} (whose interpretation picks up the most salient element). In fact both $\iota$ and $\eta$ correspond to Hilbert's $\epsilon$ 
it is only the interpretation of the two which differ. Although papers and even a book \cite{Leisenring1967epsilon} have been published on the topic, 
up to now results on these operators do not go beyond Hilbert, Ackerman and Bernays  
results in 
\cite{HBvol2}
and in particular there is not yet a sound interpretation 
that would match the natural proof theoretical rules given by Hilbert. 

and $\tau$, and others for generalised quantifiers. All those operators takes as arguments a predicate $P$ involving a free variable $x$
$P(x)$ and return a term. The $\iota$ term is written as the term $\iota x.\ P(x)$ in which the variable $x$ is bound --- the syntactical behaviour of the other generic elements introduced by $\epsilon,\tau,\eta,...$ is just the same. The main problem is to provide a proper typing of such operators which fits in our model. \footnote{Actually, we first provided a  type theoretical model,and then discovered earlier related work in untyped semantics, e.g. papers by Heusinger.} 

In a typed model, a predicate applying to $\alpha$-objects is of type $\alpha\fl\ttt$. 
Consequently $\iota$ should be of type: $(\alpha\fl\ttt)\fl\alpha$, and in order to have a single $\iota$ its type is $\Pi\alpha.\ (\alpha\fl\ttt)\fl\alpha$. Consequently, if we have a predicate \ma{Dog} of \ma{Animate}  entities the term $\iota(Dog)$ (written 
$\iota x.\ Dog(x)$ in untyped models) the semantics of \ma{the dog} is of type \ma{Animate}.... but we would like this term to be of type $Dog$ if \ma{dog} is a type,  or to enjoy the property $Dog$, if $Dog$ is a property. How do we say so, since the type $Dog$ does not appear in $\iota$? 
Indeed,  only \ma{animate} objects appear in $\iota$ as an instantiation of $\alpha$. We solve this by adding a systematic presupposition that can be called an axiom, $P(\iota(P))$ for any $P$ of type $\eee\fl\ttt$ \footnote{If the predicate $P$ corresponds to a type $\tau$ i.e.  
$P=\pred{\tau}$,   this presupposition is better written as 
$\iota(\pred{\tau}):\tau$.} 

The syntax of quantifiers and generalised quantifiers is defined in the same way. Existential quantification \ma{some} is faithfully represented by Hilbert's \emph{epsilon} operator: $P(\epsilon x P(x))\equiv \exists x.\  P(x)$. As soon as some element enjoys the property $P$, the term $\epsilon x.\ P(x)$ enjoys $P$ as well. 

The operator $\tau$ symmetrically constructs the generic element that appear in mathematical proofs like \ma{Let $x$ be any integer \ldots Thus for all integers \ldots} 
This universal generic represents universal quantification 
because $P(\tau x.\ P(x))\equiv \forall x.\ P(x)$: as soon as the term $\tau x.\ P(x)$ enjoys the property $P$ any element does. Actually,  the $\epsilon$ operator is enough, since $\tau x.\ P(x) = \epsilon x. (\lnot P(x))$ and $\epsilon x.\ P(x) =  \tau x. (\lnot P(x))$

As it is well known determiners --- at least some use of them --- correspond to quantifiers, and that's the way determiners are modelled in our framework, see e.g. \cite{Retore2013taln,Retore2013lc}. It avoids the problems evoked in examples \ref{keith} and \ref{ecp}. 

It should be observed that generics fit better into our typed and many sorted semantic representations. Indeed, intuitively it is easier to think of a generic \ma{politician} or \ma{song} than it is to think of a generic \ma{entity} or \ma{individual}.

One can even introduce constants that model generalised quantification. They are  typed just the same way, and this construct can be applied to compute the logical form of statement including the  \ma{most} 
quantifier, as exposed in \cite{Retore2012rlv}. It does not mean that we have the sound and complete proof rules nor a model theoretical  interpretation: we simply are able to automatically compute  logical forms from  sentences involving generalised quantifiers. 

\subsection{Individuals, plurals and sets in a type-theoretical framework} 

The organisation of the types also allows us to handle simple facts about plurals, 
 as shown in \cite{MootRetore2011coconat,MMR2013lenls} 
--- which resembles some Partee's ideas of \cite{partee2008np}. 
Here are some classical examples involving plurals, exemplifying some typical readings for plurals: 

\begin{exe} 
\ex. *Keith met.
\ex  Keith and John met. (unambiguous).
\ex *The student met. 
\ex  The students met. (unambiguous, one meeting)
\ex The committee met. (unambiguous, one meeting)
\ex The committees met. (ambiguous: one big meeting, one
meeting per committee, several meetings invoking several committees)
\ex The students wrote a paper. (unambiguous) 
\ex The students wrote three papers. (covering) 
\end{exe} 

Such readings are derivable in our model because one can define in \systF operators for handling plurals. 
Firstly, on can add, as a constant, a cardinality operator for predicates $||\_||:\Pi \alpha. (\alpha\fl \ttt)\fl \mathbb{N}$ (using the internal integers of system \systF which are $\mathbb{N}=\Pi X.\ (X\fl X)\fl(X\fl X)$, or predefined integers as in G\"odel system T or most type theories). 
Next, as shown in figure \ref{opplur} , we can have operators for handling plurals: $q$ (turning an individual into a property/set),  
$*$ (distributivity) $\#$ (restricted distributivity from sets of sets to its constituent subsets),  $c$ (for coverings)... The important fact is that the computation of such readings uses exactly the same mechanisms as lexical coercion. Some combinations are blocked by their types, but optional terms coming tier from the predicate or from the plural noun may allow an a priori prohibited reading. 
To be precise we also provided specific tools for handling groups that are singular nouns denoting a set.

\begin{figure} 
$\begin{array}{ll}
q & \Lambda \alpha \lambda x^\alpha \lambda y^\alpha x=y\\
{}* & \Lambda\alpha \lambda P^{\alpha\fl\ttt}\lambda Q^{\alpha\fl\ttt} \forall x^\alpha Q(x) \Rightarrow P(x)\\ 
\# & \Lambda\alpha \lambda R^{(\alpha\fl\ttt)\fl\ttt} \lambda S^{\alpha\fl\ttt)\fl\ttt}\forall P^{\alpha\fl\ttt} S(P) \Rightarrow R(P)\\ 
c & \Lambda\alpha \lambda R^{(\alpha\fl\ttt)\fl\ttt} \lambda P^{\alpha\fl\ttt} \forall x^\alpha P(x)\Rightarrow \exists Q^{\alpha\fl\ttt}
Q(x)\land (\forall y^\alpha Q(y)\Rightarrow P(y))\land R(Q)\\ 
\end{array}$
\caption{Operators for plurals} 
\label{opplur}
\end{figure}

\section{Comparison with related work and conclusion} 

\subsection{Variants and implementation} 
\label{variants} 

In the afore presented model, some points admit slight changes that do not affect the behaviour. 

As discussed in  the beginning of section \ref{ltyn} the base type can be discussed. We proposed to use  classifiers as base types of a language with classifiers, because classifiers are linguistically and  cognitively motivated classes of words and entities. But it is fairly possible that other sets of base types are better suited in particular for specific applications. \cite{MeryRetore2013nlpcs}

In relation to this issue, the inclusion between base types, 
that in our model are morphisms can be introduced with words  or as general axioms. 
We prefer the first solution which allows idiosyncratic behaviours, dependent on words
as explained  in paragraph \ref{orga} with \ma{classe} and \ma{promotion}. 
Nevertheless when dealing with ontological inclusions, or other very general coercions, 
we think a subtyping approach is possible and reduces the size of the lexicon,
this is why we are presently exploring coercive subtyping. 

The type we gave for predicate can also vary: it could be systematically $\eee\fl\ttt$, but as explained in  paragraph \ref{determiner},
types $u\fl\ttt$ are possible as well, and varying from one form to another is not complicated. 

An important variant is to define the very same ideas within a compositional model like $\lambda$-DRT \cite{musk:comb96} the compositional view of Discourse Representation Theory \cite{KR93} which can, as its name suggest, handle discursive phenomena. 
Thus one can integrate the semantical and lexical issues presented here into a broader perspective. This can be done, and in fact several applications of the model presented here are already included into the Grail parser by Richard Moot, in particular for French \cite{moot10grail}. The grammar is an automatically acquired grammar 
but unfortunately the refined semantic terms we need can only be typed by hand. 
Consequently we only tested the semantic analyses described herein on small or specific lexicon. 
For instance, four treatment of fictive motion (cf. subsection \ref{fictivemotion} has been tested with a detailed lexicon for spatial semantics, but with $\lambda$-DRT  \cite{MPR2011cid} rather than plain lambda calculus  \cite{MPR2011taln} .

\subsection{Comparison with related work} 

There are many similarities with the contemporary work by Asher and Luo described e.g. in 
\cite{AsherLuo2012sub,Luo2012lacl,ChatLuo2012cslp}. 

A first difference is the type system. Our type system, \systF,  is quite powerful but simple: four-term building operations, and two reduction rules. Luo make use of a version of Modern Type Theories (MTT), closed to the Unifying Theory of dependent Types (UTT), whose expressive power and computational complexity is difficult to compare: it is predicative but it include dependent types.  Hence it is not clear whether MTT better 
characterises the logic needed for meaning assembly. Quantification over type variable is quite comparable and allows $\forall \alpha:CN$ which is quite convenient although it can certainly be encoded within system \systF using the fact that finites sums can be defined in system \systF, hence $x:\alpha, \alpha:CN$ can be rephrased if there are finitely many $CN$  --- finite products can be fined as well. This is both a positive and negative feature of system \systF: it can encode many things, but encodings are often dull.  In addition, the MTT that Luo uses, includes dependent types, i.e. types defined from terms, which are convenient --- the way they are used so far can probably be encoded in system \systF, but encoding can be tedious. 
A possible solution, similar to \cite{Soloviev2003},  is too introduce predefined types \systF with specific reduction schemes
--- e.g. adding integers as in G\"odel's system T. 

Regarding coercions, Luo \cite{Luo2011lacl} makes an extensive use of coercive subtyping, that he introduced with Soloviev \cite{Soloviev00coercioncompletion}: as said in this paper this kind of subtyping may also work well with system \systF. So we can say that Luo system is very similar.  Dependent types, predicative quantification, may be closer to what we wish to model, 
but the formal diversity of the many rules may result in an opaque formalisation. 
The typed system at work in Asher's view \cite{Asher2011wow}  is a simple type theory extended with type constructs and operations from category theory. 
The theory extends cartesian closed category with a few of the many operations that one finds in a topos, like subtype. This approach is hardly compared with the two above, since it does not belong to the same family: morphisms do not represents (quotiented) proofs of some logic,
they are closer to a set theoretic interpretation. 

Another ingredient of  our models  are the base types. Asher leaves the set of base types open, but rather small(say a dozen) : $\eee, \ttt,$ physical object, etc., with a linguistically motivated subtyping relation $\sqsubset$ defined over these types. Luo, especially in his later article \cite{Luo2012lacl}, wants to equate base types with common nouns (also with coercions between them), and this is a possible compromise between any formula and the minimal base type system which makes it difficult to express some selectional restrictions with types. 
However it seems that they are too many of them, since not any common noun appears as a restriction of selection for another reword in a dictionary. Classifiers as base types is a recent proposal of ours which seems cognitively and linguistically motivated. It is worth exploring this hypothesis empirically in corpora and tests.

The subtyping relation between the base types are language independent in these two models, i.e. they are not triggered by words, but simply by types. We opted for a compromise in which only ontological inclusions are type driven, using coercive sub typing.

Regarding the general organisation of the lexicon and of its composition modes, the same difference applies. While according to Asher and Luo the types determine the coercions, in our approach the coercions are provided by the terms in the lexicon, i.e. by the words themselves and not by their types, with an exception for ontological inclusions.  The recent claim by Luo that base type should be common nouns (that are words) partly rubs out the differences between on one hand the  type driven approaches  of himself and Asher  and, on the other hand, ours which is more idiosyncratic being based on words and terms.  

Finally one may wonder whether we finally derive similar logical forms? They actually are quite similar: we derive higher order multi sorted logical formulae multi sorted, Asher derives formulae in an intuitionnistic set theory, which works with sorts, and Luo derives formulae of type theory. All these are more or less the same: higher order is possible although not extensively used in examples, and there are sorts or types. 

A possible difference may lie in the distance with syntax and the effective computability of the semantic representation, which requires a treatment of the current constructs in compositional semantics, like determiners, quantifiers, plurals,... and to be integrated in a general analyse also including phenomena like time or aspect. 
For the time being we did more on such issues than the others, but I am pretty sure that a similar treatment is possible within the approach  developed by Asher and Luo.  

\subsection{Perspectives} 

A part from fixing up the optimal variant among the possible variants of our model, to study and develop the convergence with related work, or to develop the implementation there are some questions both on type theory and on linguistic modelling, both theoretical and practical,  that deserve to be further studied. 

The \emph{acquisition} of the semantic lexicon has both theoretical and practical aspects. 
In particular, how could one acquire the optional lambda terms?  Syntactic informations on words can be automatically acquired, and Moot's parser that we used for experimenting our type theoretical semantic analyses was automatically acquired. \cite{moot10semi,Moot2007mit}
By now there are some techniques to acquire the usual semantic terms  of Montague semantics of \ref{Montague} that are associated with words and depicts their argument structure. \cite{Zettlemoyer2009acl}
Machine learning and serious games also apply to learn some relation between words see e.g. \cite{CimianoWenderoth2007ACL,LafourcadeJoubert2010imcsit}
But up to now there are no learning algorithms for acquiring a set of base type, nor for determining given a set of base type, the optional lambda terms, and our experiments with Moot parser were performed using hand typed semantic lexicon. 

On the logical side there are many intriguing questions. 
\begin{itemize} 
\item 
One is the relation in a type system with sorts between the (higher order) predicate calculus and the type system, exemplified by the relation between the relation between type  judgements $x:T$ that, as linguistic presuppositions,  cannot be denied and predicates $\widehat{T}(x)$ that can be denied. 
\item The Hilbert operator $\epsilon$ which look more natural in this typed system deserve to be further studied. Since most of the results are false but Hilbert's original results, the study of both the deductive system and the interpretation of those operators is appealing. 
We are especially intrigued by the formula with Hilbert operators that have no corresponding formula in usual logic. 
\item The coercive subtyping we introduced in this paper should also be further explored, e.g. by proving that there is at most one coercion between any two types. 
\item 
It is quite clear that we do not need the full power of system \systF: we chose this system of variable types and quantified types for its simplicity and elegance. Nevertheless  one may wonder whether a simple restriction that would be sufficient. Linear version of system \systF both have a lower complexity 
\cite{Lafont2004tcs}
and allow a finer grained treatment of the constraints on sense compatibility. \cite{MR2013nlcs}
\end{itemize} 

Regarding computational linguistics, and natural language processing application, the way the discourse context  
is handled, including the permanence and the propagation of constraints (e.g. on sense compatibilities) through linguistic structure. Observe that: 

\begin{exe} 
\ex 
This salmon was living nearby Scottish coast. It was delicious. 
\ex 
? This salmon that was living nearby Scottish coast was delicious. 
\ex 
* This salmon was living nearby Scottish coast and was delicious. 
\end{exe}

As a major challenge in the semantics of natural language on which this type theoretical and many sorted view might bring new lights is the semantics of mass nouns, like \emph{wine}, which can be quantified: 

\begin{exe} 
\ex He drank some wine. 
\ex He drank all the wine. 
\end{exe}

\subparagraph{Thanks} Special 
 thanks to Serge\"\i\ Soloviev for his explanations on coercive subtyping during my CNRS sabbatical at IRIT. 
Many  thanks to those I worked with on these questions R. Moot, , M. Abrusci, Ch. Bassac, B. Mery, L. Pr{\'e}vot   L. Real and to the ones I discussed with, namely   N. Asher, Z. Luo,  M. Abrusan, C. Beyssade, H. Burnett, S.-J. Conrad, F. Corblin, A. Mari, H. Person, F. del Prete.   

\bibliographystyle{plain}

\bibliography{bigbiblio} 

\begin{thebibliography}{10}

\bibitem{AbrusciRetoreCLMPS}
Vito~Michele Abrusci and Christian Retor{\'e}.
\newblock Quantification in ordinary language: from a critic of set-theoretic
  approaches to a proof-theoretic proposal.
\newblock In Peter Schr{\"o}der-Heister, editor, {\em 14th Congress of Logic,
  Methodology and Philosophy of Sciences}, 2011.

\bibitem{Asher2011wow}
Nicholas Asher.
\newblock {\em Lexical Meaning in context -- a web of words}.
\newblock Cambridge University press, 2011.

\bibitem{AsherLuo2012sub}
Nicholas Asher and Zhaohui Luo.
\newblock Formalization of coercions in lexical semantics.
\newblock In Emmanuel Chemla, Vincent Homer, and Gr{\'e}goire Winterstein,
  editors, {\em Sinn und Bedeutung 17}, pages 63--80, 2012.
\newblock \url{http://semanticsarchive.net/sub2012/}.

\bibitem{AP01}
Nicholas Asher and James Pustejovsky.
\newblock The metaphysics of words in contexts, 2000.

\bibitem{asher-typedriven}
Nicolas Asher.
\newblock A type driven theory of predication with complex types.
\newblock {\em Fundamenta Informaticae}, 84(2):151--183, 2008.

\bibitem{BMRjolli}
{C}hristian {B}assac, {B}runo {M}ery, and {C}hristian {R}etor{\'e}.
\newblock {T}owards a {T}ype-{T}heoretical {A}ccount of {L}exical {S}emantics.
\newblock {\em {J}ournal of {L}ogic {L}anguage and {I}nformation},
  19(2):229--245, April 2010.
\newblock \url{http://hal.inria.fr/inria-00408308/}.

\bibitem{LACL2012}
Denis B{\'e}chet and Alexander~Ja. Dikovsky, editors.
\newblock {\em Logical Aspects of Computational Linguistics - 7th International
  Conference, LACL 2012, Nantes, France, July 2-4, 2012. Proceedings}, volume
  7351 of {\em Lecture Notes in Computer Science}. Springer, 2012.

\bibitem{BenAviFrancez2004cg}
Gilad Ben-Avi and Nissim Francez.
\newblock Categorial grammars with ontology-refined types.
\newblock In {\em Categorial grammars -- an efficient tool for natural language
  processing}, pages 99--113, Montpellier, June 2004. C.N.R.S.

\bibitem{bierwisch1979}
Manfred Bierwisch.
\newblock W{\"o}rtliche bedeutung - eine pragmatische gretchenfrage.
\newblock In G.~Grewendorf, editor, {\em Sprechakttheorie und Semantik}, pages
  119--148. Surkamp, Frankfurt, 1979.

\bibitem{bierwisch1983}
Manfred Bierwisch.
\newblock Semantische und konzeptuelle repr{\"a}sentation lexikalischer
  einheiten.
\newblock In R.~R{\.{u}}{\u{z}}i{\u{c}}ka and W.~Motsch, editors, {\em
  Untersuchungen zur Semantik}, pages 61--99. Akademie-Verlag, Berlin, 1983.

\bibitem{blutner-pragma}
Reinhard Blutner.
\newblock Lexical semantics and pragmatics.
\newblock In Fritz Hamm and Thomas~Ede Zimmermann, editors, {\em Semantics},
  volume 10 (Sonderheft), pages 27--58, Hamburg, 2002. Buske.

\bibitem{Cardelli94Fsubtyping}
Luca Cardelli, Simone Martini, John~C. Mitchell, and Andre Scedrov.
\newblock An extension of system {F} with subtyping.
\newblock {\em Information and {C}omputation}, 109(1/2):4--56, 1994.

\bibitem{ChatLuo2012cslp}
Stergios Chatzikyriakidis and Zhaohui Luo.
\newblock An account of natural language coordination in type theory with
  coercive subtyping.
\newblock In Denys Duchier and Yannick Parmentier, editors, {\em 7th
  International Workshop on Constraint Solving and Language Processing
  (CSLP'12). Selected and Revised Papers.}, number 8114 in Lecture Notes in
  Computer Science. Springer, 2013.

\bibitem{ChatLuo2012FG}
Stergios Chatzikyriakidis and Zhaohui Luo.
\newblock Adjectives in a modern type-theoretical setting.
\newblock In Glyn Morrill and Mark-Jan Nederhof, editors, {\em FG}, volume 8036
  of {\em Lecture Notes in Computer Science}, pages 159--174. Springer, 2013.

\bibitem{CimianoWenderoth2007ACL}
Philipp Cimiano and Johanna Wenderoth.
\newblock Automatic acquisition of ranked qualia structures from the web.
\newblock In John~A. Carroll, Antal van~den Bosch, and Annie Zaenen, editors,
  {\em ACL}. The Association for Computational Linguistics, 2007.

\bibitem{cooper-codygeqlic}
Robin Cooper.
\newblock Copredication, dynamic generalized quantification and lexical
  innovation by coercion.
\newblock In {\em Fourth International Workshop on Generative Approaches to the
  Lexicon}. Universit{\'e} de Gen{\`e}ve, 2007.

\bibitem{Cooper2011lacl}
Robin Cooper.
\newblock Copredication, quantification and frames.
\newblock In Pogodalla and Prost \cite{LACL2011}, pages 64--79.

\bibitem{cruse1986lexical}
D.A. Cruse.
\newblock {\em Lexical semantics}.
\newblock Cambridge textbooks in linguistics. Cambridge University Press, 1986.

\bibitem{EgliHeusinger1995}
Urs Egli and Klaus von Heusinger.
\newblock The epsilon operator and {E}-type pronouns.
\newblock In Urs Egli, Peter~E. Pause, Christoph Schwarze, Arnim von Stechow,
  and G{\"o}tz Wienold, editors, {\em Lexical Knowledge in the Organization of
  Language}, pages 121--141. Benjamins, 1995.

\bibitem{Girard71}
Jean-Yves Girard.
\newblock Une extension de l'interpr{\'e}tation de {G}{\"o}del {\`a} l'analyse
  et son application: l'{\'e}limination des coupures dans l'analyse et la
  th{\'e}orie des types.
\newblock In Jens~Erik Fenstad, editor, {\em Proceedings of the Second
  Scandinavian Logic Symposium}, volume~63 of {\em Studies in Logic and the
  Foundations of Mathematics}, pages 63--92, Amsterdam, 1971. North Holland.

\bibitem{Girard2011blindspot}
Jean-Yves Girard.
\newblock {\em The blind spot -- lectures on logic}.
\newblock European Mathematical Society, 2011.

\bibitem{HBvol2}
David Hilbert and Paul Bernays.
\newblock {\em Grundlagen der Mathematik. Bd. 2.}
\newblock Springer, 1939.
\newblock Traduction fran{\c c}aise de F. Gaillard, E. Guillaume et M.
  Guillaume, L'Harmattan, 2001.

\bibitem{Huet76}
G{\'e}rard~P. Huet.
\newblock {\em R{\'e}solution d'{\'e}quations dans des langages d'ordre
  1,2,...,$\omega$}.
\newblock Th{\`e}se de doctorat d'{\'e}tat, Universit{\'e} Paris VII, 1976.

\bibitem{KR93}
Hans Kamp and Uwe Reyle.
\newblock {\em From Discourse to Logic}.
\newblock D. Reidel, Dordrecht, 1993.

\bibitem{Lafont2004tcs}
Yves Lafont.
\newblock Soft linear logic and polynomial time.
\newblock {\em Theoretical Computer Science}, 318(1--2):163 -- 180, 2004.
\newblock <ce:title>Implicit Computational Complexity</ce:title>.

\bibitem{LafourcadeJoubert2010imcsit}
Mathieu Lafourcade and Alain Joubert.
\newblock Computing trees of named word usages from a crowdsourced lexical
  network.
\newblock In {\em IMCSIT}, pages 439--446, 2010.

\bibitem{Lauer2004}
Sven Lauer.
\newblock A comparative study of current theories of polysemy in formal
  semantics.
\newblock Master's thesis, Cognitive science Osnabr{\"u}ck - Computational
  Linguistics, 2004.

\bibitem{LQ2012}
Alain Lecomte and Myriam Quatrini.
\newblock Figures of dialogue: a view from ludics.
\newblock {\em Synthese}, 183:59--85, 2011.

\bibitem{LMR2012cmlf}
Ana{\"\i}s Lefeuvre, Richard Moot, and Christian Retor{\'e}.
\newblock Traitement automatique d'un corpus de r{\'e}cits de voyages
  pyr{\'e}n{\'e}ens : analyse syntaxique, s{\'e}mantique et pragmatique dans le
  cadre de la th{\'e}orie des types.
\newblock In {\em Congr{\`e}s mondial de linguistique fran{\c c}aise}, 2012.

\bibitem{LMRS2012taln}
Ana{\"\i}s Lefeuvre, Richard Moot, Christian Retor{\'e}, and No{\'e}mie-Fleur
  Sandillon-Rezer.
\newblock Traitement automatique sur corpus de r{\'e}cits de voyages
  pyr{\'e}n{\'e}ens : Une analyse syntaxique, s{\'e}mantique et temporelle.
\newblock In {\em Traitement Automatique du Langage Naturel, TALN'2012},
  volume~2, pages 43--56, 2012.

\bibitem{Leisenring1967epsilon}
Albert~C. Leisenring.
\newblock {\em Mathematical logic and Hilbert's $\epsilon$ symbol}.
\newblock University Mathematical Series. Mac Donald \& Co., 1967.

\bibitem{LongoMilstedSoloviev2000lc}
Giuseppe Longo, Kathleen Milsted, and Sergei Soloviev.
\newblock Coherence and transitivity of subtyping as entailment.
\newblock {\em Journal of Logic and Computation}, 10(4):493--526, 2000.

\bibitem{Luo2011lacl}
Zhaohui Luo.
\newblock Contextual analysis of word meanings in type-theoretical semantics.
\newblock In Pogodalla and Prost \cite{LACL2011}, pages 159--174.

\bibitem{Luo2012lacl}
Zhaohui Luo.
\newblock Common nouns as types.
\newblock In B{\'e}chet and Dikovsky \cite{LACL2012}, pages 173--185.

\bibitem{LuoSolovievXue2013ic}
Zhaohui Luo, Sergei Soloviev, and Tao Xue.
\newblock Coercive subtyping: Theory and implementation.
\newblock {\em Inf. Comput.}, 223:18--42, 2013.

\bibitem{MMR2013lenls}
Bruno Mery, Richard Moot, and Christian Retor{\'e}.
\newblock Plurals: individuals and sets in a richly typed semantics.
\newblock In {\em Logic and Engineering of Natural Language Semantics 10 (LENLS
  10)}. LNCS, 2013.

\bibitem{MR2013nlcs}
Bruno Mery and Christian Retor{\'e}.
\newblock Advances in the logical representation of lexical semantics.
\newblock In Valeria de~Paiva and Larry Moss, editors, {\em Natural Language
  and Computer Science (LICS 2013 satellite workshop)}, New-Orleans, 2013.

\bibitem{MeryRetore2013nlpcs}
Bruno Mery and Christian Retor\'e.
\newblock Semantic types, lexical sorts and classifiers.
\newblock In B.~Sharp and M.~Zock, editors, {\em 10th International Workshop on
  Natural Language Processing and Cognitive Science}, Marseilles, September
  2013.

\bibitem{Moot2007mit}
Richard Moot.
\newblock Automated extraction of type-logical supertags from the spoken dutch
  corpus.
\newblock In Srinivas Bangalore and Aravind Joshi, editors, {\em The Complexity
  of Lexical Descriptions and its Relevance to Natural Language Processing: A
  Supertagging Approach}. MIT Press, 2007.

\bibitem{moot10semi}
Richard Moot.
\newblock Semi-automated extraction of a wide-coverage type-logical grammar for
  {French}.
\newblock In {\em Proceedings of Traitement Automatique des Langues Naturelles
  (TALN)}, Montreal, 2010.

\bibitem{moot10grail}
Richard Moot.
\newblock Wide-coverage {French} syntax and semantics using {Grail}.
\newblock In {\em Proceedings of Traitement Automatique des Langues Naturelles
  (TALN)}, Montreal, 2010.

\bibitem{MPR2011cid}
Richard Moot, Laurent Pr{\'e}vot, and Christian Retor{\'e}.
\newblock {A discursive analysis of itineraries in an historical and regional
  corpus of travels}.
\newblock In {\em {Constraints in discourse}}, page
  http://passage.inria.fr/cid2011/doku.php, Ayay-roches-rouges, France,
  September 2011.

\bibitem{MPR2011taln}
Richard Moot, Laurent Pr{\'e}vot, and Christian Retor{\'e}.
\newblock Un calcul de termes typ{\'e}s pour la pragmatique lexicale ---
  chemins et voyageurs fictifs dans un corpus de r{\'e}cits de voyages.
\newblock In {\em Traitement Automatique du Langage Naturel, TALN 2011}, pages
  161--166, Montpellier, France, June 2011.

\bibitem{MootRetore2011coconat}
Richard Moot and Christian Retor{\'e}.
\newblock Second order lambda calculus for meaning assembly: on the logical
  syntax of plurals.
\newblock In Reinhard Muskens, editor, {\em Coconat: Conference on Computing
  Natural Reasoning}. University of Tilburg, December 2011.
\newblock \url{http://hal.inria.fr/hal-00650644}.

\bibitem{MootRetore2012lcg}
Richard Moot and Christian Retor{\'e}.
\newblock {\em The logic of categorial grammars: a deductive account of natural
  language syntax and semantics}, volume 6850 of {\em LNCS}.
\newblock Springer, 2012.

\bibitem{Muskens91}
Reinhard Muskens.
\newblock Anaphora and the logic of change.
\newblock In Jan van Eijck, editor, {\em JELIA}, volume 478 of {\em Lecture
  Notes in Computer Science}, pages 412--427. Springer, 1990.

\bibitem{musk:comb96}
Reinhard Muskens.
\newblock Combining {M}ontague {S}emantics and {D}iscourse {R}epresentation.
\newblock {\em Linguistics and Philosophy}, 19:143--186, 1996.

\bibitem{nunberg-transfer}
Geoffrey Nunberg.
\newblock Transfers of meaning.
\newblock {\em Journal of semantics}, 12(2):109--132, 1995.

\bibitem{partee2008np}
Barbara Partee.
\newblock Noun phrase interpretation and type shifting principles.
\newblock In B.H. Partee and P.H. Portner, editors, {\em Formal Semantics: The
  Essential Readings}, pages 357--381. Wiley, 2008.

\bibitem{PetersWesterstahl2006quantifiers}
Stanley Peters and Dag Westerst{\aa}hl.
\newblock {\em Quantifiers in Language and Logic}.
\newblock Clarendon Press, 2006.

\bibitem{LACL2011}
Sylvain Pogodalla and Jean-Philippe Prost, editors.
\newblock {\em Logical Aspects of Computational Linguistics - 6th International
  Conference, LACL 2011, Montpellier, France, June 29 - July 1, 2011.
  Proceedings}, volume 6736 of {\em LNCS}. Springer, 2011.

\bibitem{Pustejovsky91}
James Pustejovsky.
\newblock The generative lexicon.
\newblock {\em Computational Linguistics}, 17(4):409--441, 1991.

\bibitem{Pus95}
James Pustejovsky.
\newblock {\em The generative lexicon}.
\newblock M.I.T. Press, 1995.

\bibitem{RealCoelhoRetore2013unilog}
Livy-Maria Real-Coelho and Christian Retor{\'e}.
\newblock A generative {M}ontagovian lexicon for polysemous deverbal nouns.
\newblock In {\em 4th World Congress and School on Universal Logic -- Workshop
  on Logic and linguistics.}, Rio de Janeiro, April 2013.

\bibitem{RealRetore2013jolli}
Livy-Maria Real-Coelho and Christian Retor{\'e}.
\newblock On the semantics of deverbals in a richly typed system.
\newblock {\em {J}ournal of {L}ogic {L}anguage and {I}nformation}, 2013.
\newblock To appear.

\bibitem{Retore2012rlv}
Christian Retor{\'e}.
\newblock Variable types for meaning assembly: a logical syntax for generic
  noun phrases introduced by "most".
\newblock {\em Recherches Linguistiques de Vincennes}, 41:83--102, 2012.

\bibitem{Retore2013lc}
Christian Retor{\'e}.
\newblock A natural framework for natural language semantics: many sorted logic
  and {H}ilbert operators in type theory.
\newblock In M\'ario Edmundo and Boban Velickovic, editors, {\em Logic
  colloquium}, Evora, 2013.

\bibitem{Retore2013taln}
Christian Retor{\'e}.
\newblock S{\'e}mantique des d{\'e}terminants dans un cadre richement typ{\'e}.
\newblock In Emmanuel Morin and Yannick Est{\`e}ve, editors, {\em Traitement
  Automatique du Langage Naturel, TALN RECITAL 2013}, volume~1, pages 367--380.
  ACL Anthology, 2013.

\bibitem{Soloviev2003}
Sergei Soloviev and David Chemouil.
\newblock {Some Algebraic Structures in Lambda-Calculus with Inductive Types}.
\newblock In Stefano Berardi, Mario Coppo, and Ferruccio Damiani, editors, {\em
  TYPES}, volume 3085 of {\em Lecture Notes in Computer Science}, pages
  338--354. Springer, 2003.

\bibitem{Soloviev00coercioncompletion}
Sergei Soloviev and Zhaohui Luo.
\newblock Coercion completion and conservativity in coercive subtyping.
\newblock {\em Annals of Pure and Applied Logic}, 1-3(113):297--322, 2000.

\bibitem{talmy99fictive}
Leonard Talmy.
\newblock Fictive motion in language and ``ception''.
\newblock In Paul Bloom, Mary~A. Peterson, Lynn Nadel, and Merrill~F. Garrett,
  editors, {\em Language and Space}, pages 211--276. {MIT} Press, 1999.

\bibitem{Heusinger1997}
Klaus von Heusinger.
\newblock Definite descriptions and choice functions.
\newblock In S.~Akama, editor, {\em Logic, Language and Computation}, pages
  61--91. Kluwer, 1997.

\bibitem{Heusinger2004}
Klaus von Heusinger.
\newblock Choice functions and the anaphoric semantics of definite nps.
\newblock {\em Research on Language and Computation}, 2:309--329, 2004.

\bibitem{LuoXue2012lacl}
Tao Xue and Zhaohui Luo.
\newblock Dot-types and their implementation.
\newblock In B{\'e}chet and Dikovsky \cite{LACL2012}, pages 234--249.

\bibitem{Zettlemoyer2009acl}
Luke~S. Zettlemoyer and Michael Collins.
\newblock Learning context-dependent mappings from sentences to logical form.
\newblock In Keh-Yih Su, Jian Su, and Janyce Wiebe, editors, {\em ACL/IJCNLP},
  pages 976--984. The Association for Computer Linguistics, 2009.

\end{thebibliography}


\end{document}